\documentclass[submission,copyright,creativecommons]{eptcs}
\providecommand{\event}{EXPRESS/SOS 2022}

\usepackage{underscore}
\usepackage[T1]{fontenc}
\usepackage{xspace}
\usepackage{stmaryrd, amsmath, amssymb, amsthm}
\usepackage{nicefrac}
\usepackage{tikz}
	\usetikzlibrary{trees,decorations,arrows,automata,positioning,plotmarks,
	shapes,backgrounds,petri}
\usepackage{wrapfig}
\usepackage{cancel}

\usepackage{macros}

\title{On the Expressiveness of Mixed Choice Sessions}
\author{Kirstin Peters
	\institute{Universität Augsburg\\ Germany}
	\email{kirstin.peters@uni-a.de}
\and
	Nobuko Yoshida
	\institute{Imperial College London\\ UK}
	\email{n.yoshida@imperial.ac.uk}
}
\def\titlerunning{On the Expressiveness of Mixed Choice Sessions}
\def\authorrunning{K. Peters and N. Yoshida}

\begin{document}

\maketitle


\begin{abstract}
Session types provide a flexible programming style for structuring
interaction, and are used to guarantee a safe and consistent
composition of distributed processes. Traditional session types 
include only one-directional input (external) and output (internal)
guarded choices. This prevents the session-processes 
to explore the full expressive power  
of the \piCal where the mixed choices are proved more expressive 
than the (non-mixed) guarded choices. To account this issue,   
recently Casal, Mordido, and Vasconcelos proposed 
the binary session types with mixed choices (\CMVmix).
This paper carries a surprising, unfortunate result on \CMVmix:
in spite of an inclusion of unrestricted channels with mixed choice, 
\CMVmix's mixed choice is rather separate and not 
mixed. We prove this negative result using 
two methodologies (using either the leader election problem or a synchronisation pattern as distinguishing feature), showing that there exists no good encoding 
from the \piCal into \CMVmix, preserving distribution.  
We then close their open problem on  
the encoding from \CMVmix into 
\CMV (without mixed choice), proving 
its soundness and thereby that the encoding is good up to coupled similarity.  
\end{abstract}


\section{Introduction}
\label{sec:introduction}

Starting with the landmark result by Palamidessi in \cite{palamidessi97} and followed up by results such as \cite{nestmann00, palamidessi03, gorla10, peters12, petersNestmann12, breakingSymmetries16} it was shown that the key to the expressive power of the full \piCal in comparison to its sub-calculi such as \eg the asynchronous \piCal is \emph{mixed choice}.

\emph{Mixed choice} in the \piCal is a choice construct that allows to choose between inputs and outputs.
In contrast, \eg \emph{separate choices} are constructed from either only inputs or only outputs.
The additional expressive power of mixed choice relies on its ability to rule out alternative options of the opposite nature, \ie a term can rule out its possibility to perform an input by doing an output, whereas without mixed choice inputs can rule out alternative inputs only and outputs may rule out only alternative outputs.

To compare calculi with different variants of choice, we try to build an encoding or show that no such encoding exists \cite{boerPalamidessi1991, parrow08}.
The existence of an encoding that satisfies relevant criteria shows that the target language is expressive enough to emulate the behaviours in the source language.
Gorla \cite{gorla10} and others \cite{parrow08, petersNestmannGoltz13} introduced and
classified a set of general criteria
for encodability which are syntax-agnostic 
\cite{gorla10, petersNestmannGoltz13}: 
they are now commonly used for claiming expressiveness of a given
calculus, defining important features which a ``good encoding'' should
satisfy. These include \emph{compositionality} (homomorphism), 
\emph{name invariance} (bijectional renaming), 
sound and complete \emph{operational correspondence} (the source and
target can simulate each other), 
\emph{divergence reflection} (the target diverges only if the
source diverges), \emph{observability} (barb-sensitiveness), and \emph{distributability} preservation (the target has the same degree of distribution as the source). 
Conversely, a \emph{separation result}, \ie the proof of the absence of an encoding with certain criteria, shows that the source language can represent behaviours that can\emph{not} be expressed in the target. This paper gives a fresh look at
expressiveness of typed $\pi$-calculi,
focusing on choice constructs of \emph{session types}. 

Session types \cite{hondaVasconcelosKubo98,THK} specify and constrain the communication behaviour as a
protocol between components in a system. A session type system
excludes any non-conforming behaviour, statically preventing 
type and communication errors (i.e., mismatch of choice labels).
Several languages now have session-type support via libraries and
tools \cite{BETTYTOOLBOOK,DBLP:journals/ftpl/AnconaBB0CDGGGH16}. 
As the origin of session types is Linear Logic \cite{HondaK:typdyi}, 
traditional session types 
include only one-directional input (external) and output (internal)
guarded choices. To explore the full expressiveness of mixed choice
from the \piCal, 
recently Casal, Mordido, and Vasconcelos proposed 
the binary session types with mixed choices called \emph{mixed sessions}
\cite{casalMordidoVasconcelos22}.
We denote their calculus by \CMVmix. 
Mixed sessions include a mixture of branchings (labelled input choices) and 
selections (labelled output choices) 
at the same \emph{linear} channel or \emph{unrestricted} channel.
This extension gives us many useful and typable \emph{structured}
concurrent programming idioms which consist
of both unrestricted and linear non-deterministic choice behaviours.
We show that
in spite of its practical relevance,
 mixed sessions in \CMVmix 
are \emph{strictly less expressive} than
mixed choice in the \piCal
even with an unrestricted usage of choice channels.

This result surprised us.
We would have expected that using mixed choice with an unrestricted choice channel results into a choice construct comparable to choice in the \piCal.
But, as we show in the following, mixed choice in \CMVmix cannot express essential features of mixed choice in the \piCal.
First we observe that mixed sessions are not expressive enough to solve leader election in symmetric networks.
Remember that it was leader election in symmetric networks that was used to show that mixed choice is more expressive than separate choice in the \piCal (see \cite{palamidessi97}).
Second we observe that mixed sessions cannot express the synchronisation pattern \patternStar.
Synchronisation patterns were introduced in \cite{petersNestmannGoltz13} to capture the amount of synchronisation that can be expressed in distributed systems.
The synchronisation pattern \patternStar was identified in \cite{petersNestmannGoltz13} as capturing exactly the amount of synchronisation introduced with mixed choice in the \piCal.
Finally, we have a closer look at the encoding from \CMVmix into \CMV presented in \cite{casalMordidoVasconcelos22}.
\CMV is the variant of session types that is extended in \cite{casalMordidoVasconcelos22} with a mixed-choice-construct in order to obtain \CMVmix, \ie \CMV has traditional branching and selection but not their mix.
As it is the case for many variants of session types, \CMV can express separate choice but has no construct for mixed choice.
By analysing this encoding, we underpin our claim that mixed choice in \CMVmix is not more expressive than separate choice in the \piCal.

\begin{wrapfigure}{R}{0.35\textwidth}
	\centering
	\begin{tikzpicture}[bend angle=20]
		\node (pi) at (2, 2) {$ \pi $};
		\node (CMV) at (4, 0) {\CMV};
		\node (CMVmix) at (2, 0) {\CMVmix};
		\path[->, thick, color=darkgreen] (CMVmix) edge (CMV);
		\path[->, color=red, dashed] (pi) edge [bend right] node[left] {$ \mathsf{LE} $\;} node {$ \times $} (CMVmix);
		\path[->, color=red, dashed] (pi) edge [bend left] node[right] {\;\patternStar} node {$ \times $} (CMVmix);
	\end{tikzpicture}
\end{wrapfigure}

Our contributions are summarised in the picture on the right.
In \S~\ref{sec:separateMixedSessionsLeaderElection} we prove that there exists no good encoding from the \piCal (with mixed choice)
into \CMVmix, where we use the leader election problem by Palamidessi in \cite{palamidessi97} (\textcolor{red}{$ \mathsf{LE} $}) as distinguishing feature (the first \begin{tikzpicture}[baseline=-1mm] \path[->, dashed, color=red] (0, 0) edge node {$ \times $} (1, 0); \end{tikzpicture}).
In \S~\ref{sec:separateMixedSessionsFromPiSynchronisationPatterns} we reprove this result using the \emph{synchronisation pattern} \textcolor{red}{\patternStar} from \cite{petersNestmannGoltz13} instead as distinguishing feature (the second \begin{tikzpicture}[baseline=-1mm] \path[->, dashed, color=red] (0, 0) edge node {$ \times $} (1, 0); \end{tikzpicture}).
Then we prove soundness of the encoding presented in
\cite{casalMordidoVasconcelos22} closing their open problem in
\S~\ref{sec:encodeMixedSessions} (\begin{tikzpicture}[baseline=-1mm]
  \path[->, thick, color=darkgreen] (0, 0) edge (0.5, 0); \end{tikzpicture}).
By this encoding source terms in \CMVmix and their literal
translations in \CMV are related by \emph{coupled similarity} \cite{parrow1992}, \ie \CMVmix is encoded into \CMV up to coupled similarity.
From the separation results in \S~\ref{sec:separateMixedSessionsLeaderElection} and \S~\ref{sec:separateMixedSessionsFromPiSynchronisationPatterns} and the encoding into session types with separate choice in \S~\ref{sec:encodeMixedSessions} we conclude that \emph{mixed sessions in \cite{casalMordidoVasconcelos22} can express only separate choice}.

To make our paper readable, we focus on presenting our results with intuitive, self-contained explanations.
In particular, we simplify the languages \CMVmix and \CMV for the discussion in this paper and omit their type systems.
However, the proofs are carried out on the original definitions of the languages from \cite{casalMordidoVasconcelos22}. 
We include complete proofs of all the statements in this paper and the technical details of the notions from the literature that we use in \cite{petersYoshidaTecRep22}.


\section{Technical Preliminaries: Mixed Sessions and Encodability Criteria}
\label{sec:preliminaries}

This section gives a summary of the \piCal, \CMVmix, \CMV, and encodability criteria.

Assume a countably-infinite set $ \names $ of \emph{names}.
For the \piCal we additionally assume a set $ \Set{ \overline{y} \mid y \in \names } $ of \emph{co-names}.
Let $ \tau \notin \names \cup \Set{ \overline{y} \mid y \in \names } $.

The \emph{syntax} of a process calculus is usually defined by a context-free grammar defining operators.
We use $ P, Q, \ldots $ to range over process terms.
The arguments of a term $ P $ that are again process terms are called \emph{subterms} of $ P $.
Terms that appear as subterm underneath some (action) prefix are called \emph{guarded}, because the guarded subterm cannot be executed before the guarding action has been performed.
Also conditionals, such as if-then-else-constructs, guard their respective subterms.
\emph{Expressions} are constructed from variables, unit, and standard boolean operators.
We assume an evaluation function $ \Eval{\cdot} $ that evaluates expressions to \emph{values}.

We assume that the \emph{semantics} is given as an \emph{operational semantics} consisting of inference rules defined on the operators of the language \cite{Plotkin04}. For many process calculi, the semantics is provided in two forms, as \emph{reduction semantics} and as \emph{labelled transition semantics}. We assume that at least the reduction semantics $ \step $ is given as part of the definition, because its treatment is easier in the context of encodings.
A \emph{(reduction) step} is written as $ P \step P' $. If $ P \step P' $, then $ P' $ is called \emph{derivative} of $ P $. Let $ P \step $ (or $ P \noStep $) denote the existence (absence) of a step from $ P $, and let $ \steps $ denote the reflexive and transitive closure of $ \step $. A sequence of reduction steps is called a \emph{reduction}.
We write $ P \step^{\omega} $ if $ P $ has an infinite sequence of steps.
We also use \emph{execution} to refer to a reduction starting from a particular term.
A process that cannot reduce is called \emph{stuck}.

The application $ P\sigma $ of a substitution $ \sigma = \Set{ \Subst{y_1}{x_1}, \ldots, \Subst{y_n}{x_n} } $ on a term is defined as the result of simultaneously replacing all free occurrences of $ x_i $ by $ y_i $ for $ i \in \Set{ 1, \ldots, n } $, possibly applying $ \alpha $-conversion to avoid capture or name clashes.
For all names in $ \names \setminus \Set{ x_1, \ldots, x_n } $ the substitution behaves as the identity mapping.

\subsection{Process and Session Calculi}

The \piCal was introduced by Milner, Parrow, and Walker in \cite{MilnerParrowWalker92} and is one of the most well-known process calculi.
We consider a variant of the \piCal with mixed guarded choice and replication but without matching (as in \cite{palamidessi97}).
This variant is often called the synchronous or full \piCal.
\emph{Mixed sessions} are a variant of binary session types introduced by Casal, Mordido, and Vasconcelos in \cite{casalMordidoVasconcelos22} with a choice construct that combines prefixes for sending and receiving.
We denote this language as \CMVmix.
\CMV is the session type variant on which \CMVmix is based.
A central idea of \CMVmix (and \CMV) is that channels are separated in two \emph{channel endpoints} and that interaction is by two processes acting on the respective different ends of such a channel.

The syntax of the \piCal, \CMVmix, and \CMV is given as:
\begin{align*}
	\procPi\text{: } & P \deffTerms \sum_{i \in \indexSet} \alpha_i.P_i \sepTerms \ResPi{x}{P} \sepTerms P \mid P \sepTerms !P
	\hspace{2em} \alpha \deffTerms \InpPi{y}{x} \sepTerms \OutPi{y}{z} \sepTerms \tau\\
	\procCMVmix\text{: } & P \deffTerms y \, \sum_{i \in \indexSet}M_i \sepTerms P \mid P \sepTerms \ResCMVmix{y}{z}{P} \sepTerms \ConditionalCMVmix{v}{P}{P} \sepTerms \inactCMVmix
	\hspace{2em} M \deffTerms \BranchCMVmix{\Label}{*}{v}{P}
	\hspace{2em} * \deffTerms ! \sepTerms ?\\
	\procCMV\text{: } & P \deffTerms \OutCMV{y}{v}{P} \sepTerms {y}{?}{x}{P} \sepTerms \SelCMV{x}{\Label}{P} \sepTerms \BranCMV{x}{\Set{\BranchCMV{\Label_i}{P_i}}_{i \in \indexSet}} \sepTerms P \mid P \sepTerms \ResCMVmix{y}{z}{P} \sepTerms \ConditionalCMVmix{v}{P}{P} \sepTerms \inactCMVmix
\end{align*}

\noindent
A choice $ \sum_{i \in \indexSet} \alpha_i.P_i $ in $ \procPi $ offers for each $ i $ in the index set $ \indexSet $ a subterm guarded by some action prefix $ \alpha_i $, where $ \alpha_i $ is an input action $ \InpPi{y}{x} $, an output action $ \OutPi{y}{z} $, or the internal action $ \tau $.
In contrast choice $ y \, \sum_{i \in \indexSet}M_i $ in \CMVmix operates on a single channel endpoint.
In \cite{casalMordidoVasconcelos22} a choice is declared as either linear ($ \linCMVmix $) or unrestricted ($ \unCMVmix $), where the latter introduces recursion in the calculus.
We omit these qualifiers to simplify the presentation.
However, the proofs are carried out on the languages \CMVmix and \CMV as given by \cite{casalMordidoVasconcelos22} (see \cite{petersYoshidaTecRep22}).
A branch $ \BranchCMVmix{\Label}{*}{v}{P} $ specifies a label $ \Label $, a polarity $ * $ ($ ! $ for sending or $ ? $ for receiving), a value in output actions or a variable for input actions, and a continuation $ P $.
We abbreviate the empty sum by $ \inactCMVmix$ and 
we often separate summands by $ + $.
The remaining are: restriction $ \ResPi{x}{P} $ and $ \ResCMV{y}{z}{P} $, parallel composition $ P \mid P $, replication $ !P $, conditional $ \ConditionalCMVmix{v}{P}{P} $, output $ \OutCMV{y}{v}{P} $, input $ {y}{?}{x}{P} $, selection $ \SelCMV{x}{\Label}{P} $, and branching $ \BranCMV{x}{\Set{\BranchCMV{\Label_i}{P_i}}_{i \in \indexSet}} $.

The semantics of the languages is given by the axioms:
\begin{displaymath}\begin{array}{lc}
	\pi\text{: } & \OutPi{y}{z}.P + R \mid \InpPi{y}{x}.Q + N \stepPi P \mid Q\Set{\Subst{z}{x}} \hspace{2em} \tau.P + R \stepPi P
	\vspace{0.5em}\\
	\CMVmix/\CMV\text{: } & \ConditionalCMVmix{\true}{P}{Q} \stepCMVmix P \hspace{2em} \ConditionalCMVmix{\false}{P}{Q} \stepCMVmix Q
	\vspace{0.5em}\\
	\CMVmix\text{: } & \ResCMVmix{y}{z}{{\left( y \, {\left( \OutCMVmix{\Label}{v}{P} + M \right)} \mid z \, {\left( \InpCMVmix{\Label}{x}{Q} + N \right)} \mid R \right)}} \stepCMVmix \ResCMVmix{y}{z}{{\left( P \mid Q\Set{\Subst{v}{x}} \mid R \right)}}
	\vspace{0.5em}\\
	\CMV\text{: } & \ResCMV{y}{z}{{\left( \OutCMV{y}{v}{P} \mid {z}{?}{x}{Q} \mid R \right)}} \stepCMV \ResCMV{y}{z}{{\left( P \mid Q\Set{\Subst{v}{x}} \mid R \right)}}
	\vspace{0.5em}\\
	& \ResCMV{y}{z}{{\left( \SelCMV{y}{\Label_j}{P} \mid \BranCMV{z}{\Set{\BranchCMV{\Label_i}{Q_i}}_{i \in \indexSet}} \mid R \right)}} \stepCMV \ResCMV{y}{z}{{\left( P \mid Q_j \mid R \right)}}
\end{array}\end{displaymath}
and all three calculi share the following rules:
\begin{displaymath}\begin{array}{c}
	\dfrac{P \stepPi P'}{P \mid Q \stepPi P' \mid Q}
	\hspace{2em}
	\dfrac{P \stepPi P'}{\ResPi{\hat{x}}{P} \stepPi \ResPi{\hat{x}}{P'}}
	\hspace{2em}
	\dfrac{P \scPi Q \quad Q \stepPi Q' \quad Q' \scPi P'}{P \stepPi P'}
\end{array}\end{displaymath}
where $ \hat{x} $ consists of either one or two names; and $\scPi$ denotes the standard \emph{structural congruence} from \cite{casalMordidoVasconcelos22} plus a rule for replication in the \piCal.
More precisely, structural congruence is defined as the least congruence that contains $\alpha$-conversion and satisfies the rules:
\begin{displaymath}\begin{array}{c}
	\ResPi{\hat{x}}{\inactPi} \scPi \inactPi
	\hspace{2em}
	\ResPi{\hat{x}}{\ResPi{\hat{y}}{P}} \scPi \ResPi{\hat{y}}{\ResPi{\hat{x}}{P}}
	\hspace{2em}
	P \mid \ResPi{\hat{x}}{Q} \scPi \ResPi{\hat{x}}{{\left( P \mid Q \right)}} \quad \text{if } \Set{\hat{x}} \cap \FreeNames{P} = \emptyset
	\vspace{0.5em}\\
	\ResCMVmix{y}{z}{P} \scCMVmix \ResCMVmix{z}{y}{P}
	\hspace{1em}
	P \mid \inactPi \scPi P
	\hspace{1em}
	P \mid Q \scPi Q \mid P
	\hspace{1em}
	P \mid {\left( Q \mid R \right)} \scPi {\left( P \mid Q \right)} \mid R
	\hspace{1em}
	!P \scPi P \mid {!P}
\end{array}\end{displaymath}

The name $ x $ is bound in $ P $ by $ \InpPi{y}{x}.P $, $ \BranchCMVmix{\Label}{?}{x}{P} $, and $ \ResPi{x}{P} $, $ \ResCMVmix{x}{y}{P} $, or $ \ResCMVmix{y}{x}{P} $. All other names are free.
We often omit $ \inactPi $ and 
the argument of action prefixes if it is irrelevant, \ie we write $ y.P $ instead of $ \InpPi{y}{x}.P $ if $ x \notin \FreeNames{P} $; and $ \overline{y}.P $ instead of $ \OutPi{y}{z}.P $ if for all matching receivers $ \InpPi{y}{x}.Q $ we have $ x \notin \FreeNames{Q}$.

A process $ P $ emits a barb $ \overline{y} $, denoted as 
$ P\Barb{\overline{y}} $, if $ P $ (in the \piCal) has an unguarded output $ \OutPi{y}{z}.P' $ on a free channel $ y \in \FreeNames{P} $.
Similarly, $ P $ has a barb $ y $, denoted as $ P\Barb{y} $, if $ P $ (in the \piCal) has an unguarded input $ \InpPi{y}{x}.P' $ or if $ P $ (in \CMVmix/\CMV) has an unguarded choice $ \ChoiceCMVmix{}{y}{\sum_{i \in \indexSet}M_i} $, output $ \OutCMV{y}{v}{P} $, input $ \InpCMV{}{y}{x}{P} $, selection $ \SelCMV{y}{\Label}{P} $, or branching $ \BranCMV{y}{\Set{\BranchCMV{\Label_i}{P_i}}_{i \in \indexSet}} $ on a free $ y \in \FreeNames{P} $.
In \CMVmix and \CMV we do not distinguish between input barbs $ \Barb{y} $ and output barbs $ \Barb{\overline{y}} $ but instead have barbs on different channel end points.
The term $ P $ reaches a barb $ \beta $ (with $ \beta = y $ or $ \beta = \overline{y} $), denoted as $ P\WeakBarb{\beta} $, if there is some $ P' $ such that $ P \steps P' $ and $ P'\Barb{\beta} $.

The type systems of \CMVmix and \CMV in \cite{casalMordidoVasconcelos22} ensure that the two endpoints of a channel are dual, \eg if one endpoint sends the other has to receive.
For the expressiveness results on the choice construct proved in this paper, the type system is not crucial.
Indeed, our separation result (in both ways to prove it) is carried out on the untyped version of $ \CMVmix $.
See \cite{casalMordidoVasconcelos22} or \cite{petersYoshidaTecRep22} for the full typing systems.

Two terms of a language are usually compared using some kind of a behavioural simulation relation.
The most commonly known behavioural simulation relation is bisimulation.
A relation $ \mathcal{R} $ is a bisimulation if any two related processes mutually simulate their respective sequences of steps, such that the derivatives are again related.

\newpage
\begin{definition}[Bisimulation]
	$ \mathcal{R} $ is a (weak reduction, barbed) bisimulation if for each $ {\left( P, Q \right)} \in \mathcal{R} $:
	\begin{itemize}
		\item $ P \steps P' $ implies $ \exists Q'\logdot Q \steps Q' \wedge {\left( P', Q' \right)} \in \mathcal{R} $
		\item $ Q \steps Q' $ implies $ \exists P'\logdot P \steps P' \wedge {\left( P', Q' \right)} \in \mathcal{R} $
		\item $ P\WeakBarb{\beta} $ iff $ Q\WeakBarb{\beta} $ for all barbs $ \beta $
	\end{itemize}
	Two terms are bisimilar if there exists a bisimulation that relates them.
	For a language $ \lang $, let $ \approx_{\lang} $ denote bisimilarity on $ \lang $.
\end{definition}

Another interesting behavioural simulation relation is coupled similarity.
It was introduced in \cite{parrow1992} and discussed \eg in \cite{bispingNestmannPeters19}.
It is strictly weaker than bisimilarity.
As pointed out in \cite{parrow1992}, in contrast to bisimilarity it essentially allows for intermediate states (see \S~\ref{sec:encodeMixedSessions}).
Each symmetric coupled simulation is a bisimulation.

\begin{definition}[Coupled Simulation]
	A relation $ \mathcal{R} $ is a (weak reduction, barbed) coupled simulation if for each $ {\left( P, Q \right)} \in \mathcal{R} $:
	\begin{itemize}
		\item $ P \steps P' $ implies $ \exists Q'\logdot Q \steps Q' \wedge {\left( P', Q' \right)} \in \mathcal{R} $
		\item $ P \steps P' $ also implies $ \exists Q'\logdot Q \steps Q' \wedge {\left( Q', P' \right)} \in
\mathcal{R} $ 
		\item $ P\WeakBarb{\beta} $ implies $ Q\WeakBarb{\beta} $ for all barbs $ \beta $
	\end{itemize}
	Two terms are coupled similar if they are related by a coupled simulation in both directions.
\end{definition}

For all languages considered here, a process $ P $ is distributable into $ P_1, \ldots, P_n $ if and only if we have $ P \equiv \ResPi{\tilde{y}}{\left( P_1 \mid \ldots \mid P_n \right)} $ (compare to the notion of a standard form of the \piCal in \cite{Milner1999} and the discussion on distributability in \cite{petersNestmannGoltz13}).
Moreover, two steps $ a: P \step P_a $ and $ b : P \step P_b $ of $ P $ are in conflict if one disables the other, \ie if $ a $ and $ b $ compete for some action prefix.
More precisely, two steps in the \piCal are in conflict if they reduce the same choice, two steps in \CMVmix are in conflict if they reduce the same choice or the same conditional, and two steps in \CMV are are in conflict if the reduce the same output, input, selection prefix, branching prefix, or conditional.
Note that reducing the same choice not necessarily means to reduce the same summand in this choice.
The steps $ a $ and $ b $ are distributable, if $ P $ is distributable into at least two parts such that one part performs the step $ a $ and the other part performs $ b $.

\subsection{Encodability Criteria}

An encoding $\arbitraryEncoding $ is a function from the processes of the source language into the processes of the target language, where we need encodability criteria to rule out trivial or meaningless encodings.
In order to provide a general framework, Gorla in \cite{gorla10} suggests five criteria well suited for language comparison.
Other frameworks were introduced \eg in \cite{fu16, glabbeek18}.
We replace success sensitiveness of \cite{gorla10} by the stricter
barb-sensitiveness, because it is more intuitive.
As we claim, all separation results in this paper remain valid \wrt success sensitiveness instead of barb-sensitiveness.

The papers \cite{casalMordidoVasconcelos22} and \cite{palamidessi97} require as additional criterion that the parallel operator is translated homomorphically.
To strengthen the separation results, we use the slightly weaker criterion `preservation of distributability' (see \cite{petersNestmann12, peters12}).
The encoding of \cite{casalMordidoVasconcelos22} that we discuss in \S~\ref{sec:encodeMixedSessions} translates the parallel operator homomorphically.

\newpage

\begin{definition}[Good Encoding]
	\label{def:goodEncodingA}
	We consider an encoding $ \arbitraryEncoding $ to be \emph{good} if it is
	\begin{description}
		\item[compositional:] The translation of an operator is captured by a context that takes as arguments the translations of the subterms of the operator.
		\item[name invariant:] For every $ S $ and every substitution $ \sigma $, it holds that $ \ArbitraryEncoding{S\sigma} \asymp \ArbitraryEncoding{S}\sigma $.
		\item[operationally complete:] For all $ S \stepsSource S' $, it holds $ \ArbitraryEncoding{S} \stepsTarget \asymp \ArbitraryEncoding{S'} $.
		\item[operationally sound:] For all $ \ArbitraryEncoding{S} \stepsTarget T $, there is an $ S' $ s.t.\ $ S \stepsSource S' $ and $ T \stepsTarget \asymp \ArbitraryEncoding{S'} $.
		\item[divergence reflecting:] For every $ S $, $ \ArbitraryEncoding{S} \stepTarget^{\omega} $ implies $ S \stepSource^{\omega} $.
		\item[barb-sensitive:] For every $ S $ and every barb $ y $, $ S\WeakBarb{y} $ iff $ \ArbitraryEncoding{S}\WeakBarb{y} $.
		\item[distributability preserving:] For every $ S \in \procSource $ and for all terms $ S_1, \ldots, S_n \in \procSource $ that are distributable within $ S $ there are some $ T_1, \ldots, T_n \in \procTarget $ that are distributable within $ \ArbitraryEncoding{S} $ such that $ T_i \asymp \ArbitraryEncoding{S_i} $ for all $ 1 \leq i \leq n $.
	\end{description}
	Moreover the equivalence $ \asymp $ is a barb respecting (weak) reduction bisimulation.
\end{definition}

Operational correspondence is the combination of operational \emph{completeness} and \emph{soundness}.
Since we are focusing on separation results on untyped languages, we do not require an explicit criterion for types.
In \cite{casalMordidoVasconcelos22} the encoding from \CMVmix into \CMV is proven to be type sound.


\section{Separating Mixed Sessions and the Pi-Calculus via Leader Election}
\label{sec:separateMixedSessionsLeaderElection}

The first expressiveness result on the \piCal that focuses on mixed choice is the separation result by Palamidessi in \cite{palamidessi97, palamidessi03}.
This result uses the problem of leader election in symmetric networks as distinguishing feature.

Following \cite{palamidessi97} we assume that the set of names $ \names $ contains names that identify the processes of the network and that are never used as bound names within electoral systems.
For simplicity, we use natural numbers for this kind of names.
A leader is announced by unguarding an output on its id.
Then a network $ P = \ResPi{\tilde{x}}{{\left( P_1 \mid \ldots \mid P_k \right)}} $ in $ \procPi $ or $ P = \ResCMVmix{\tilde{x}}{\tilde{y}}{{\left( P_1 \mid \ldots \mid P_k \right)}} $ in $ \procCMVmix $ is an \emph{electoral system} if in every maximal execution exactly one leader is announced.
We adapt the definition of electoral systems of \cite{palamidessi97} to obtain electoral systems in the \piCal and in \CMVmix.

\begin{definition}[Electoral System]
	A network $ P = \ResPi{\tilde{x}}{{\left( P_1 \mid \ldots \mid P_k \right)}} $ in $ \procPi $ or $ P = \ResCMVmix{\tilde{x}}{\tilde{y}}{{\left( P_1 \mid \ldots \mid P_k \right)}} $ in $ \procCMVmix $ is an \emph{electoral system} if for every execution $ E: P \steps P' $ there exists an extension $ E': P \steps P' \steps P'' $ and some $ n \in \Set{ 1, \ldots, k } $ (the leader) such that $ P'''\Barb{n} $ for all $ P''' $ with $ P'' \steps P''' $, but $ P''\NotWeakBarb{m} $ for any $ m \in \Set{1, \ldots, k} $ with $ m \neq n $.
\end{definition}

Accordingly, an electoral system in the \piCal announces a leader by unguarding some output on $ n $ that cannot be reduced or removed, where $ n $ is the id of the leader.
In \CMVmix a leader is announced by unguarding a choice on the channel $ n $.
Since $ n $ is free this choice cannot be removed.
A network is an electoral system if in every maximal execution exactly one leader $ n $ is announced.

We adapt the definition of hypergraphs that are associated to a network of processes in the \piCal defined in \cite{palamidessi97} to networks in \CMVmix.
The hypergraph connects the nodes $ 1, \ldots, k $ of the network by edges representing the free channels that they share, where we ignore the outer restrictions of the network.

\begin{definition}[Hypergraph]
	Given a network $ P = \ResPi{\tilde{x}}{{\left( P_1 \mid \ldots \mid P_k \right)}} $ in $ \procPi $ or $ P = \ResCMVmix{\tilde{x}}{\tilde{y}}{{\left( P_1 \mid \ldots \mid P_k \right)}} $ in $ \procCMVmix $, the \emph{hypergraph} associated to $ P $ is $ \Hypergraph{P} = \Tuple{N, X, t} $ with $ N = \Set{ 1, \ldots, k } $, $ X = \FreeNames{P_1 \mid \ldots \mid P_n} \setminus N $, and $ t(x) = \Set{n \mid x \in \FreeNames{P_n}} $ for each $ x \in X $.
\end{definition}

Because we ignore the outer restrictions of the network in the above definition, the hypergraphs of two structural congruent networks may be different.
However, this is not crucial for our results.

Given a hypergraph $ H = \Tuple{N, X, t} $, an automorphism on $ H $ is a pair $ \sigma = \Tuple{\sigma_N, \sigma_X} $ such that $ \sigma_N: N \to N $ and $ \sigma_X: X \to X $ are permutations which preserve the type of arcs.
For simplicity, we usually do not distinguish between $ \sigma_N $ and $ \sigma_X $ and simply write $ \sigma $.
Moreover, since $ \sigma $ is a substitution, we allow to apply $ \sigma $ on terms $ P $, denoted as $ P\sigma $.
The orbit $ \Orbit{\sigma}{n} $ of $ n \in N $ generated by $ \sigma $ is defined as the set of nodes in which the various iterations of $ \sigma $ map $ n $, \ie $ \Orbit{\sigma}{n} = \Set{ n, \sigma(n), \ldots, \sigma^{h - 1}(n) } $, where $ \sigma^i $ represents the composition of $ \sigma $ with itself $ i $ times and $ \sigma^h = \id $.
We also adapt the notion of a symmetric system of \cite{palamidessi97} to obtain symmetric systems in the \piCal as well as in \CMVmix.

\begin{definition}[Symmetric System]
	Consider a network $ P = \ResPi{\tilde{x}}{{\left( P_1 \mid \ldots \mid P_k \right)}} $ in $ \procPi $ or a network $ P = \ResCMVmix{\tilde{x}}{\tilde{y}}{{\left( P_1 \mid \ldots \mid P_k \right)}} $ in $ \procCMVmix $, and let $ \sigma $ be an isomorphism on its associated hypergraph $ \Hypergraph{P} = \Tuple{N, X, t} $. $ P $ is \emph{symmetric \wrt $ \sigma $} iff $ P_{\sigma(i)} \approx_{\pi} P_i\sigma $ or $ P_{\sigma(i)} \approx_{\CMVmix} P_i\sigma $ for each node $ i \in N $. $ P $ is \emph{symmetric} if it is symmetric \wrt all the automorphisms of $ \Hypergraph{P} $.
\end{definition}

\noindent
In contrast to \cite{palamidessi97} we use bisimilarity---$ \approx_{\pi} $ and $ \approx_{\CMVmix} $---instead of alpha conversion in the definition of symmetry.
With this weaker notion of symmetry, we compensate for the weaker criterion on distributability that we use instead of the homomorphic translation of the parallel operator.
Accordingly, we also consider networks as symmetric if they behave in a symmetric way; they do not necessarily need to be structurally symmetric.

In the \piCal we find symmetric electoral systems for many kinds of hypergraphs.
We use such a solution of leader election in a network with five nodes as counterexample to separate \CMVmix from the \piCal.

\begin{example}[Leader Election in the \PiCal]
	\label{exa:leaderElectionPi}
	Consider the network
	\begin{align*}
		\LEPi &= \ResPi{a, b, c, d, e, v, w, x, y, z}{\left( S_1 \mid S_2 \mid S_3 \mid S_4 \mid S_5 \right)} 
	\end{align*}
	where $ S_1 = \overline{e} + a.{\left( \overline{x} + v.\overline{1} \right)} $, $ S_2 = \overline{a} + b.{\left( \overline{y} + w.\overline{2} \right)} $, $ S_3 = \overline{b} + c.{\left( \overline{z} + x.\overline{3} \right)} $, $ S_4 = \overline{c} + d.{\left( \overline{v} + y.\overline{4} \right)} $, and $ S_5 = \overline{d} + e.{\left( \overline{w} + z.\overline{5} \right)} $.
	\qed
\end{example}

\begin{wrapfigure}{R}{0.275\textwidth}
	\centering
	\scalebox{0.8}{
	\begin{tikzpicture}[bend angle=20]
		\foreach \v/\x/\y/\z in {2/1/a/v,1/2/b/w,5/3/c/x,4/4/d/y,3/5/e/z}
        {
            \path (360*\v/5-55:1.75) node[draw, inner sep=0pt, minimum size=3pt] (p\x) {$ \begin{array}{c} \x\\ \textcolor{blue}{\y} \;\; \textcolor{red}{\z} \end{array} $};
        }
        \draw[-latex, color=blue] (p1) edge [bend right] node[above] {$ \overline{e} $} (p5);
        \draw[-latex, color=blue] (p2) edge [bend right] node[above] {$ \overline{a} $} (p1);
        \draw[-latex, color=blue] (p3) edge [bend right] node[right] {$ \overline{b} $} (p2);
        \draw[-latex, color=blue] (p4) edge [bend right] node[below] {$ \overline{c} $} (p3);
        \draw[-latex, color=blue] (p5) edge [bend right] node[left] {$ \overline{d} $} (p4);
        \draw[-latex, color=red] (p1) edge node[pos=0.1, right] {$ \overline{x} $} (p3);
        \draw[-latex, color=red] (p2) edge node[pos=0.1, below] {$ \overline{y} $} (p4);
        \draw[-latex, color=red] (p3) edge node[pos=0.1, below] {$ \overline{z} $} (p5);
        \draw[-latex, color=red] (p4) edge node[pos=0.1, left] {$ \overline{v} $} (p1);
        \draw[-latex, color=red] (p5) edge node[pos=0.1, above] {$ \overline{w} $} (p2);
	\end{tikzpicture}
	}
\end{wrapfigure}

$ \LEPi $ is symmetric.
Consider \eg the permutation $ \sigma $ that permutes the channels as follows: $ a \rightarrow b \rightarrow c \rightarrow d \rightarrow e \rightarrow a $, $ v \rightarrow w \rightarrow x \rightarrow y \rightarrow z \rightarrow v $, and $ 1 \rightarrow 2 \rightarrow 3 \rightarrow 4 \rightarrow 5 \rightarrow 1 $.
Then $ S_{\sigma(i)} = S_i\sigma $ for all $ i \in \Set{ 1, \ldots, 5 } $.
The network elects a leader in two stages.
The first stage (depicted as \textcolor{blue}{blue circle}) uses mixed choices on the channels $ a, b, c, d, e $; in the second stage (depicted as a \textcolor{red}{red star}) we have mixed choices on the channels $ v, w, x, y, z $.
The picture on the right gives $ \Hypergraph{\LEPi} $ extended by arrow heads to visualise the direction of interactions and the respective action prefixes.
The senders in the two stages are losing the leader election game, \ie are not becoming the leader.
In the first stage two processes can be receivers and continue with the second stage.
The process that is neither sender nor receiver in the first stage is stuck and also loses.
The receiver of the second stage then becomes the leader by unguarding an output on its id.
For instance we obtain the execution
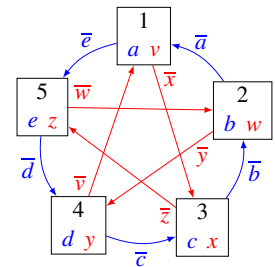
\begin{align*}
	\LEPi &\stepPi \ResPi{\tilde{n}}{{\left( \overline{x} + v.\overline{1} \mid S_3 \mid S_4 \mid S_5 \right)}} \stepPi \ResPi{\tilde{n}}{{\left( \overline{x} + v.\overline{1} \mid \overline{z} + x.\overline{3} \mid S_5 \right)}} \stepPi \overline{3} \mid \ResPi{\tilde{n}}{S_5} \noStep
\end{align*}
with $ \tilde{n} = a, b, c, d, e, v, w, x, y, z $ by reducing on the channels $ a $ and $ c $ in the first stage.
The network $ \LEPi $ has 10 maximal executions (modulo structural congruence) that are obtained from the above execution by symmetry on the first two steps.
In each maximal execution exactly one leader is elected.

We show that there exists no symmetric electoral system for networks of
size five in \CMVmix; or more generally no symmetric electoral system
for networks of odd size in \CMVmix.
A key ingredient to separate the \piCal with mixed choice from the asynchronous \piCal in \cite{palamidessi97} is a confluence lemma.
It states that in the asynchronous \piCal a step reducing an output and an alternative step reducing an input cannot be conflict to each other and thus can be executed in any order.
In the full \piCal this confluence lemma is not valid, because inputs and outputs can be combined within a single choice construct and can thus be in conflict.
For \CMVmix we observe that steps that reduce different endpoints can also not be in conflict to each other, because different channel endpoints cannot be combined in a single choice.

\begin{wrapfigure}{R}{0.2\textwidth}
	\centering
	\scalebox{0.8}{
	\begin{tikzpicture}[bend angle=20]
		\node (a) at (0, 0.75) {$ A $};
		\node (b) at (1.5, 1.5) {$ B $};
		\node (c) at (1.5, 0) {$ C $};
		\node (d) at (3, 0.75) {$ D $};
		\path[|->] (a) edge (b);
		\path[|->] (a) edge (c);
		\path[|->] (b) edge (d);
		\path[|->] (c) edge (d);
	\end{tikzpicture}
	}
\end{wrapfigure}

\begin{lemma}[Confluence]
	\label{lem:confluenceCMVmix}
	Let $ P, Q \in \procCMVmix $.
	Assume that $ A = \ResCMVmix{\tilde{x}}{\tilde{y}}{{\left( P \mid Q \right)}} $ can make two steps $ A \step \ResCMVmix{\widetilde{x_1}}{\widetilde{y_1}}{{\left( P_1 \mid Q_1 \right)}} = B $ and $ A \step \ResCMVmix{\widetilde{x_2}}{\widetilde{y_2}}{{\left( P_2 \mid Q_2 \right)}} = C $ such that $ P_1 $ is obtained modulo $ \scCMVmix $ from $ P $ by reducing a choice on channel endpoint $ a $ and $ P_2 $ is obtained modulo $ \scCMVmix $ from $ P $ by reducing a choice on channel endpoint $ b $ with $ a \neq b $.
	Then there exist $ P_3, Q_3 \in \procCMVmix $ and $ D = \ResCMVmix{\widetilde{x_3}}{\widetilde{y_3}}{{\left( P_3 \mid Q_3 \right)}} $ such that $ B \step D $ and $ C \step D $, where $ \widetilde{x_3} = \widetilde{x_1} \cup \widetilde{x_2} $ and $ \widetilde{y_3} = \widetilde{y_1} \cup \widetilde{y_2} $.
\end{lemma}

The proof of this confluence lemma relies on the observation that the two steps of $ A $ to $ B $ and $ C $ have to reduce distributable parts of $ A $.
Then these two steps are distributable, which in turn allows us to perform them in any order.
Thus the expressive power of choice in \CMVmix is limited by the fact that syntactically the choice construct is fixed on a single channel endpoint.
With this alternative confluence lemma, we can show that there is no electoral system of odd degree in \CMVmix.

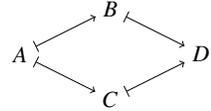
\begin{lemma}[No Electoral System]
	Consider a network $ P = \ResCMVmix{\tilde{x}}{\tilde{y}}{{\left( P_1 \mid \ldots \mid P_k \right)}} $ in \CMVmix with $ k > 1 $ being an odd number. Assume that the associated hypergraph $ \Hypergraph{P} $ admits an automorphism $ \sigma \neq \id $ with only one orbit, and that $ P $ is symmetric \wrt $ \sigma $. Then $ P $ cannot be an electoral system.
	\label{lem:noElectoralSystemCMVmix}
\end{lemma}

In the proof we construct a potentially infinite sequence of steps such that the system constantly restores symmetry, \ie whenever a step destroys symmetry we can perform a sequence of steps that restores the symmetry.
Therefore we rely on the assumption of $ \sigma $ generating only one orbit.
This implies that $ \Orbit{\sigma}{i} = \Set{ i, \sigma(i), \ldots, \sigma^{k - 1}(i) } = \Set{ 1, \ldots, k } $, for each $ i \in \Set{ 1, \ldots, k } $.
Because of that, whenever part $ i $ performs a step that destroys symmetry or parts $ i $ and $ j $ together perform a step that destroys symmetry, the respective other parts of the originally symmetric network can perform symmetric steps to restore the symmetry of the network.
Because of the symmetry, the constructed sequence of steps does not elect a unique leader.
Accordingly, the existence of this sequence ensures that $ P $ is not an electoral system.
The odd degree of the network is necessary to ensure that we can apply the confluence lemma, which in turn ensures that we can always perform the sequence of steps to restore symmetry after the step that destroys the symmetry.

By the preservation of distributability, encodings preserve the
structure of networks; and 
by name invariance, they also preserve the symmetry of networks.
With operational correspondence and barb-sensitiveness, any good encoding of $ \LEPi $ is again a symmetric electoral system of size five.
Since by Lemma~\ref{lem:noElectoralSystemCMVmix} this is not possible, 
we can separate \CMVmix from the \piCal.

\begin{theorem}[Separate \CMVmix from the \PiCal via Leader Election]
	\label{thm:separateCMVmixfromPiviaLeaderElection}
	$ $\\
	There is no good encoding from the \piCal into \CMVmix.
\end{theorem}


\newpage
\section{Separating Mixed Sessions and the Pi-Calculus via Synchronisation}
\label{sec:separateMixedSessionsFromPiSynchronisationPatterns}

\begin{wrapfigure}{R}{0.25\textwidth}
	\centering
	\scalebox{0.6}{
	\tikzstyle{place}=[circle,draw=black,thick,minimum size=5mm]
	\tikzstyle{transition}=[rectangle,draw=black,thick,minimum size=5mm]
	\begin{tikzpicture}
		\foreach \x/\xtext in {1/e,2/d,3/c,4/b,5/a}
        {
            \path (360*\x/5+125:0.8) node[transition] (\xtext) {$\xtext$};
            \path (360*\x/5-55:1.75) node[place,tokens=1] (p\x) {};
        }

        \draw[-latex] (p2) -- (a);
        \draw[-latex] (p2) -- (b);

        \draw[-latex] (p1) -- (b);
        \draw[-latex] (p1) -- (c);

        \draw[-latex] (p5) -- (c);
        \draw[-latex] (p5) -- (d);

        \draw[-latex] (p4) -- (d);
        \draw[-latex] (p4) -- (e);

        \draw[-latex] (p3) -- (e);
        \draw[-latex] (p3) -- (a);
	\end{tikzpicture}
	}
\end{wrapfigure}

In \cite{petersNestmannGoltz13} the technique used in \cite{palamidessi97} and its relation to synchronisation are analysed.
Two synchronisation patterns, the pattern \patternM and the pattern \patternStar, are identified that describe two different levels of synchronisation and allow to more clearly separate languages along their ability to express synchronisation.
These patterns are called \patternM and \patternStar, because their respective representations as a Petri net (see left and right picture) have these shapes.
The pattern \patternStar captures the power of synchronisation of the \piCal.
In particular it captures what is necessary to solve the leader election problem.

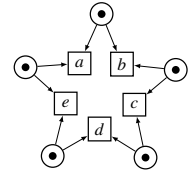
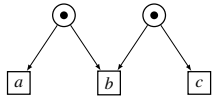
\begin{wrapfigure}{L}{0.25\textwidth}
	\centering
	\scalebox{0.6}{
	\tikzstyle{place}=[circle,draw=black,thick,minimum size=5mm]
	\tikzstyle{transition}=[rectangle,draw=black,thick,minimum size=5mm]
	\begin{tikzpicture}
		\node[place,tokens=1]	(p) at (1, 1.5) {};
		\node[place,tokens=1]	(q) at (3, 1.5) {};
		\node[transition]		(a) at (0, 0) {$ a $};
		\node[transition]		(b) at (2, 0) {$ b $};
		\node[transition]		(c) at (4, 0) {$ c $};

		\draw[-latex] (p) -- (a);
		\draw[-latex] (p) -- (b);
		\draw[-latex] (q) -- (b);
		\draw[-latex] (q) -- (c);
	\end{tikzpicture}
	}
\end{wrapfigure}

The pattern \patternM captures a very weak form of synchronisation, not enough to solve leader election but enough to make a fully distributed implementation of languages with this pattern difficult (see also \cite{petersNestmannIC20}).
This pattern was originally identified in \cite{Glabbeek2008} when studying the relevance of synchrony and distribution on Petri nets.
As shown in \cite{peters12, petersNestmannGoltz13}, the ability to express these different amounts of synchronisation in the \piCal lies in its different forms of choices: to express the pattern \patternStar the \piCal needs mixed choice, whereas separate choice allows to express the pattern \patternM.
Indeed we find the pattern \patternM in \CMVmix, but there are no \patternStar in \CMVmix.

\begin{example}[A \patternM in \CMVmix]
	The process $ \PMCMVmix $ is a \patternM in \CMVmix:
	\begin{center}
		\begin{tikzpicture}[]
			\node (res) at (-0.4, 1) {$ \PMCMVmix = \ResCMVmix{x}{y}{(} $};
			\node[fill=blue!20, rounded corners=0.5cm] (l1) at (3, 1) {$ \begin{array}{c} \text{\begin{tiny} location 1 \end{tiny}}\\ \ChoiceCMVmix{}{x}{{\left( \OutCMVmix{\Label}{\true}{P_1} + \InpCMVmix{\Label}{z}{P_2} \right)}}\\ \ChoiceCMVmix{}{y}{{\left( \InpCMVmix{\Label}{z}{P_5} + \OutCMVmix{\Label}{\true}{P_6} \right)}} \end{array} $};
			\node (p1) at (5.25, 1) {$ \mid $};
			\node (p2) at (5.25, 0.55) {$ \mid $};
			\node[fill=green!20, rounded corners=0.5cm] (l2) at (7.5, 1) {$ \begin{array}{c} \text{\begin{tiny} location 2 \end{tiny}}\\ \ChoiceCMVmix{}{x}{{\left( \OutCMVmix{\Label}{\false}{P_3} + \InpCMVmix{\Label}{z}{P_4} \right)}}\\ \ChoiceCMVmix{}{y}{{\left( \InpCMVmix{\Label}{z}{P_7} + \OutCMVmix{\Label}{\false}{P_8} \right)}} \end{array} $};
			\node (p3) at (9.75, 1) {$ \mid $};
			\node (b) at (10, 0.55) {$ ) $};
		\end{tikzpicture}
	\end{center}
	A process is a \patternM if it can perform three steps $ a, b, c $, where $ a, b, c $ are names and not labels, such that $ a $ and $ b $ as well as $ b $ and $ c $ are in conflict whereas $ a $ and $ c $ are distributable steps.
	For instance we can pick the steps $ a $, $ b $, and $ c $ as:
	\begin{description}
		\item[Step $ a $:] $ \PMCMVmix \step \ResCMVmix{x}{y}{\left( P_1 \mid \ChoiceCMVmix{}{x}{{\left( \OutCMVmix{\Label}{\false}{P_3} + \InpCMVmix{\Label}{z}{P_4} \right)}} \mid P_5\Set{\Subst{\true}{z}} \mid \ChoiceCMVmix{}{y}{{\left( \InpCMVmix{\Label}{z}{P_7} + \OutCMVmix{\Label}{\false}{P_8} \right)}} \right)} $
		\item[Step $ b $:] $ \PMCMVmix \step \ResCMVmix{x}{y}{\left( P_1 \mid \ChoiceCMVmix{}{x}{{\left( \OutCMVmix{\Label}{\false}{P_3} + \InpCMVmix{\Label}{z}{P_4} \right)}} \mid \ChoiceCMVmix{}{y}{{\left( \InpCMVmix{\Label}{z}{P_5} + \OutCMVmix{\Label}{\true}{P_6} \right)}} \mid P_7\Set{\Subst{\true}{z}} \right)} $
		\item[Step $ c $:] $ \PMCMVmix \step \ResCMVmix{x}{y}{\left( \ChoiceCMVmix{}{x}{{\left( \OutCMVmix{\Label}{\true}{P_1} + \InpCMVmix{\Label}{z}{P_2} \right)}} \mid P_3 \mid \ChoiceCMVmix{}{y}{{\left( \InpCMVmix{\Label}{z}{P_5} + \OutCMVmix{\Label}{\true}{P_6} \right)}} \mid P_7\Set{\Subst{\false}{z}} \right)} $
	\end{description}
	The process $ \PMCMVmix $ is well-typed (see \cite{petersYoshidaTecRep22}). \qed
\end{example}

We use synchronisation patterns and the proof technique presented in \cite{petersNestmannGoltz13} to present an alternative way to prove Theorem~\ref{thm:separateCMVmixfromPiviaLeaderElection}.
By that we underpin our claim that the choice construct of \CMVmix is separate and not mixed, and we provide further intuition on why this choice construct is less expressive.

We inherit the definition of the synchronisation pattern \patternStar from \cite{petersNestmannGoltz13}, where we do not distinguish between local and non-local \patternStar since in the \piCal there is no difference between parallel and distributable steps.

\begin{definition}[Synchronisation Pattern \patternStar]
	\label{def:synchronisationPatternGreatM}
Let $ \ProcCal{\proc}{\step} $ be a process calculus and $ \PS \in \proc $ such that:
	\begin{itemize}
		\item $ \PS $ can perform at least five alternative reduction steps $ i : \PS \step P_i $ for $ i \in \Set{ a, b, c, d, e } $ such that the $ P_i $ are pairwise different;
		\item the steps $ a $, $ b $, $ c $, $ d $, and $ e $
                  form a circle such that $ a $ is in conflict with $
                  b $, $ b $ is in conflict with $ c $, $ c $ is in
                  conflict with $ d $, $ d $ is in conflict with $ e
                  $, and $ e $ is in conflict with $ a $; and 
		\item every pair of steps in $ \Set{ a, b, c, d, e } $ that is not in conflict due to the previous condition is distributable in $ \PS $.
	\end{itemize}
	In this case, we denote the process $ \PS $ as \patternStar.
\end{definition}

In contrast to \CMVmix we do find \patternStar in the \piCal. 

\begin{example}[The \patternStar in the \PiCal]
	\label{exa:piStar}
	Consider the following \patternStar in the \piCal:
	\begin{align*}
		\PSPi = \overline{a}{} + b.\overline{o_b} \mid \overline{b} + c.\overline{o_c} \mid \overline{c} + d.\overline{o_d} \mid \overline{d} + e.\overline{o_e} \mid \overline{e} + a.\overline{o_a} 
	\end{align*}
	The steps $ a, \ldots, e $ of Definition~\ref{def:synchronisationPatternGreatM} are the steps on the respective channels.
	\begin{description}
		\item[Step $ a $:] $ \PSPi \step S_a $ with $ S_a = \OutPi{b}{} + \InpPi{c}{}.\OutPi{o_c}{} \mid \OutPi{c}{} + \InpPi{d}{}.\OutPi{o_d}{} \mid \OutPi{d}{} + \InpPi{e}{}.\OutPi{o_e}{} \mid \OutPi{o_a}{} $,
		\item[Step $ b $:] $ \PSPi \step S_b $ with $ S_b = \OutPi{o_b}{} \mid \OutPi{c}{} + \InpPi{d}{}.\OutPi{o_d}{} \mid \OutPi{d}{} + \InpPi{e}{}.\OutPi{o_e}{} \mid \OutPi{e}{} + \InpPi{a}{}.\OutPi{o_a}{} $,
		\item[Step $ c $:] $ \PSPi \step S_c $ with $ S_c = \OutPi{a}{} + \InpPi{b}{}.\OutPi{o_b}{} \mid \OutPi{o_c}{} \mid \OutPi{d}{} + \InpPi{e}{}.\OutPi{o_e}{} \mid \OutPi{e}{} + \InpPi{a}{}.\OutPi{o_a}{} $,
		\item[Step $ d $:] $ \PSPi \step S_d $ with $ S_d = \OutPi{a}{} + \InpPi{b}{}.\OutPi{o_b}{} \mid \OutPi{b}{} + \InpPi{c}{}.\OutPi{o_c}{} \mid \OutPi{o_d}{} \mid \OutPi{e}{} + \InpPi{a}{}.\OutPi{o_a}{} $
		\item[Step $ e $:] $ \PSPi \step S_e $ with $ S_e = \OutPi{a}{} + \InpPi{b}{}.\OutPi{o_b}{} \mid \OutPi{b}{} + \InpPi{c}{}.\OutPi{o_c}{} \mid \OutPi{c}{} + \InpPi{d}{}.\OutPi{o_d}{} \mid \OutPi{o_e}{} $
	\end{description}
	The different outputs $ \OutPi{o_x} $ allow to distinguish between the different steps by their observables.\qed
\end{example}

We use the \patternStar $ \PSPi $ as counterexample to show that there is no good encoding from the \piCal into \CMVmix.
From Lemma~\ref{lem:noElectoralSystemCMVmix} we learned that \CMVmix cannot express certain electoral systems.
Accordingly, we are not surprised that \CMVmix cannot express the pattern \patternStar.

\begin{lemma}
	There are no \patternStar in \CMVmix.
	\label{lem:noStarCMVmix}
\end{lemma}

\begin{proof}
	Assume the contrary, \ie assume that there is a term $ \PSCMVmix $ in \CMVmix that is a \patternStar.
	Then $ \PSCMVmix $ can perform at least five alternative reduction steps $ a, b, c, d, e $ such that neighbouring steps in the sequence $ a, b, c, d, e, a $ are pairwise in conflict and non-neighbouring steps are distributable.
	Since steps reducing a conditional cannot be in conflict with any other step, none of the steps in $ \Set{a, b, c, d, e} $ reduces a conditional.
	Then all steps in $ \Set{a, b, c, d, e} $ are communication steps that reduce an output and an input that both are part of choices (with at least one summand).
	Because of the conflict between $ a $ and $ b $, these two steps reduce the same choice but this choice is not reduced in $ c $, because $ a $ and $ c $ are distributable.

	\vspace{0.25em}
	\noindent
	\begin{minipage}{\textwidth}
	\begin{wrapfigure}{R}{0.225\textwidth}
		\centering
		\scalebox{0.8}{
		\begin{tikzpicture}[]
			\foreach \x/\xlabel/\xtext/\ytext in {1/e/$ C_5 $/$ b $,2/d/$ C_4 $/$ a $,3/c/$ C_3 $/$ e $,4/b/$ C_2 $/$ d $,5/a/$ C_1 $/$ c $}
	        {
	            \path (360*\x/5+125:0.8) node (\xlabel) {\xtext};
	            \path (360*\x/5-55:1.75) node (p\x) {\ytext};
	        }

	        \draw[-latex] (p2) -- (a);
	        \draw[-latex] (p2) -- (b);

	        \draw[-latex] (p1) -- (b);
	        \draw[-latex] (p1) -- (c);

	        \draw[-latex] (p5) -- (c);
	        \draw[-latex] (p5) -- (d);

	        \draw[-latex] (p4) -- (d);
	        \draw[-latex] (p4) -- (e);

	        \draw[-latex] (p3) -- (e);
	        \draw[-latex] (p3) -- (a);
		\end{tikzpicture}
		}
	\end{wrapfigure}
	
	\hspace{1em}
	By repeating this argument, we conclude that in the steps $ a, b, c, d, e $ five choices $ C_1, \ldots, C_5 $ are reduced as depicted on the right,
	where \eg the step $ a $ reduces the choices $ C_1 $ and $ C_2 $.
	By the reduction semantics of \CMVmix, the two choices $ C_1 $ and $ C_2 $ that are reduced in step $ a $ need to use dual endpoints of the same channel.
	Without loss of generality, assume that $ C_1 $ is on channel endpoint $ x $ and $ C_2 $ is on channel endpoint $ y $.
	Then the choice $ C_3 $ needs to be on channel endpoint $ x $ again, because step $ b $ reduces $ C_2 $ (on $ y $) and $ C_3 $.
	By repeating this argument, then $ C_4 $ is on $ y $ and $ C_5 $ is on $ x $.
	But then step $ e $ reduces two choices $ C_1 $ and $ C_5 $ that are both on channel endpoint $ x $.
	Since the reduction semantics of \CMVmix does not allow such a step, this is a contradiction.
	\end{minipage}

	We conclude that there are no \patternStar in \CMVmix.
\end{proof}

The proof of the above lemma tells us more about why choice in \CMVmix is limited.
From the confluence property in \CMVmix we get the hint that the problem is the restriction of choice to a single channel endpoint.
A \patternStar is a circle of steps of odd degree, where neighbouring steps are in conflict.
More precisely, the star with five points in \patternStar is the smallest cycle of steps where neighbouring steps are in conflict and that contains non-neighbouring distributable steps.
The proof shows that the limitation of choice to a single channel endpoint and the requirement of the semantics that a channel endpoint can interact with exactly one other channel endpoint causes the problem.
This also explains why Lemma~\ref{lem:noElectoralSystemCMVmix} considers electoral systems of odd degree, because the odd degree does not allow to close the cycle as explained in the proof above.
Indeed, if we change the syntax to allow mixed choice with summands on more than one channel, we obtain the mixed-choice-construct of the \piCal.
Similarly, we invalidate our separation result in the Theorems~\ref{thm:separateCMVmixfromPiviaLeaderElection} and \ref{thm:separateCMVmixfromPiviaStar}, if we change the semantics to allow two choices to communicate even if they are on the same channel.
The latter may be more surprising, but indeed we do not need more than a single channel to solve leader election and build \patternStar, \eg $ \PSPi $ remains a star if we choose $ a = b = c = d = e $ (though we might want to pick different names $ o_a, \ldots, o_e $ to be able to distinguish the steps).

We use $ \PSPi $ in Example~\ref{exa:piStar} as counterexample.

\begin{theorem}[Separate \CMVmix and the \PiCal via \patternStar]
	\label{thm:separateCMVmixfromPiviaStar}
	$ $\\
	There is no good and distributability preserving encoding from the \piCal into \CMVmix.
\end{theorem}

To prove the above theorem, we show that the conflicts in the counterexample $ \PSPi $ have to be translated into conflicts in its literal translation.
Since the target language \CMVmix cannot express a \patternStar, to emulate $ \PSPi $ it has to break the cycle and split at least one of the conflicts with the respective two neighbouring steps into two distributable conflicts: one for each neighbour.
This causes a contradiction, because the distribution of the conflict induces new behaviour that is observable modulo the criteria we picked for good encodings.


\section{Encoding Mixed Sessions into Separate Choice}
\label{sec:encodeMixedSessions}

In \cite[\S~7]{casalMordidoVasconcelos22} an encoding of mixed sessions (\CMVmix) into the variant of this session type system \CMV with only separate choice (branching and selection) is presented.
The proof of soundness of this encoding is missing in \cite{casalMordidoVasconcelos22}.
They suggest to prove soundness modulo ``a weak form of bisimulation''.
As discussed below, the soundness criterion used in \cite{casalMordidoVasconcelos22} needs to be corrected first.

The main idea of $ \EncCMVmixCMV{\cdot} $ is to encode the information about whether a summand is an output or an input into the label used in branching, where a label $ \Label_i $ used with polarity $ ! $ in a choice typed as internal becomes $ \Label_{i, !} $ and in a choice typed as external it becomes $ \Label_{i, ?} $.
The dual treatment of polarities \wrt the type ensures that the labels of matching communication partners are translated to the same label.

\begin{example}[Translation]
	\label{exa:translation}
	Consider for example the term $ S \in \procCMVmix $:
	\begin{align*}
		S &= \ResCMVmix{x}{y}{\left( \ChoiceCMVmix{}{y}{\left( \textcolor{blue}{\OutCMVmix{\Label}{\false}{S_1}} + \textcolor{darkgreen}{\InpCMVmix{\Label}{z}{S_2}} \right)} \mid \ChoiceCMVmix{}{x}{\left( \textcolor{orange}{\OutCMVmix{\Label}{\true}{\inactCMVmix}} + \textcolor{red}{\InpCMVmix{\Label}{z}{\inactCMVmix}} \right)} \mid \ChoiceCMVmix{}{y}{\left( \OutCMVmix{\Label}{\false}{S_3} + \InpCMVmix{\Label}{z}{S_4} \right)} \right)}
	\end{align*}
	$ S $ is well-typed but the type system forces us to assign dual types to $ x $ and $ y $.
	Because of that, the choices on one channel need to be internal and on the other external.
	Let us assume that we have external choices on $ y $ and that the choice on $ x $ is internal.
	Moreover, we assume that both channels are marked as linear but typed as unrestricted.
	Then the translation\footnote{Note that \cite{casalMordidoVasconcelos22} introduces a typed encoding, thus $ \EncCMVmixCMV{P} $ actually means $ \EncCMVmixCMV{\Gamma \vdash P} $, where $ \Gamma \vdash P $ is the type statement ensuring that $ P $ is well-typed.} yields $ \EncCMVmixCMV{S} \steps T_1 $ with
	\begin{align*}
		T_1 = \ResCMV{x}{y}{\big(
			\!\!\begin{array}[t]{l}
				\InpCMV{\!}{y}{c}{\BranCMV{c}{\Set{\BranchCMV{\textcolor{blue}{\Label_{?}}}{\left( \OutCMV{c}{\textcolor{blue}{\false}}{\textcolor{blue}{\EncCMVmixCMV{S_1}}} \mid J_1 \right)}, \quad \BranchCMV{\textcolor{darkgreen}{\Label_{!}}}{\left( \InpCMV{\!}{c}{\textcolor{darkgreen}{z}}{\textcolor{darkgreen}{\EncCMVmixCMV{S_2}}} \mid J_2 \right)}}}}\\
				{}\mid \ResCMV{s}{t}{\big(
					\!\!\begin{array}[t]{l}
						\BranCMV{s}{\Set{ \BranchCMV{\Label_1}{\ResCMV{c}{d}{\left( \OutCMV{x}{c}{\SelCMV{d}{\textcolor{orange}{\Label_{!}}}{\left( \OutCMV{d}{\textcolor{orange}{\true}}{\textcolor{orange}{\inactCMV}} \mid J_3\right)}} \right)}}, \quad \BranchCMV{\Label_2}{\ResCMV{c}{d}{\left( \OutCMV{x}{c}{\SelCMV{d}{\textcolor{red}{\Label_{?}}}{\left( \InpCMV{\!}{d}{\textcolor{red}{z}}{\textcolor{red}{\inactCMV}} \mid J_4\right)}} \right)}} }}\\
						{}\mid \SelCMV{t}{\Label_1}{\inactCMV} \mid \SelCMV{t}{\Label_2}{\inactCMV} \big)
					\end{array}}\\
				{}\mid \InpCMV{\!}{y}{c}{\BranCMV{c}{\Set{\BranchCMV{\Label_{?}}{\left( \OutCMV{c}{\false}{\EncCMVmixCMV{S_3}} \mid J_5 \right)}, \quad \BranchCMV{\Label_{!}}{\left( \InpCMV{\!}{c}{z}{\EncCMVmixCMV{S_4}} \mid J_6 \right)}}}} \big)
			\end{array}}
	\end{align*}
	where we already performed a few steps to hide some technical details of the encoding function $ \EncCMVmixCMV{\cdot} $ that are not relevant for this explanation and where the $ J_1, \ldots, J_6 $ remain as junk from performing these steps.
	We call terms junk if they are stuck and do not emit barbs, \ie we can ignore the junk.
	In particular, junk is invisible modulo $ \bisimCMV $.
	We observe, that in the translation of the first $ \ChoiceCMVmix{}{y}{\left( \textcolor{blue}{\OutCMVmix{\Label}{\false}{S_1}} + \textcolor{darkgreen}{\InpCMVmix{\Label}{z}{S_2}} \right)} $ in the first line of $ T_1 $ the output with label $ \textcolor{blue}{\Label} $ is translated to the label $ \textcolor{blue}{\Label_{?}} $ and the input with label $ \textcolor{darkgreen}{\Label} $ is translated to the label $ \textcolor{darkgreen}{\Label_{!}} $, whereas in the translation of its dual $ \ChoiceCMVmix{}{x}{\left( \textcolor{orange}{\OutCMVmix{\Label}{\true}{\inactCMVmix}} + \textcolor{red}{\InpCMVmix{\Label}{z}{\inactCMVmix}} \right)} $ in the second line of $ T_1 $ we obtain $ \textcolor{orange}{\Label_{!}} $ for the output and $ \textcolor{red}{\Label_{?}} $ for the input.
	To emulate the step $ S \step S_2' = \ResCMVmix{x}{y}{\left( S_2\Set{\Subst{\true}{z}} \mid \ChoiceCMVmix{}{y}{\left( \OutCMVmix{\Label}{\false}{S_3} + \InpCMVmix{\Label}{z}{S_4} \right)} \right)} $ of $ S $ in that $ \textcolor{orange}{\true} $ is transmitted to $ \textcolor{darkgreen}{S_2} $, we start by picking the corresponding alternative, namely $ \Label_1 $ for sending, in the second and third line of $ T_1 $
	\begin{align*}
		T_1 \step T_2 = \ResCMV{x}{y}{\big(
			\!\!\begin{array}[t]{l}
				\InpCMV{\!}{y}{c}{\BranCMV{c}{\Set{\BranchCMV{\textcolor{blue}{\Label_{?}}}{\left( \OutCMV{c}{\textcolor{blue}{\false}}{\textcolor{blue}{\EncCMVmixCMV{S_1}}} \mid J_1 \right)}, \quad \BranchCMV{\textcolor{darkgreen}{\Label_{!}}}{\left( \InpCMV{\!}{c}{\textcolor{darkgreen}{z}}{\textcolor{darkgreen}{\EncCMVmixCMV{S_2}}} \mid J_2 \right)}}}}\\
				{}\mid \ResCMV{c}{d}{\left( \OutCMV{x}{c}{\SelCMV{d}{\textcolor{orange}{\Label_{!}}}{\left( \OutCMV{d}{\textcolor{orange}{\true}}{\textcolor{orange}{\inactCMV}} \mid J_3\right)}} \right)} \mid J_7\\
				{}\mid \InpCMV{\!}{y}{c}{\BranCMV{c}{\Set{\BranchCMV{\Label_{?}}{\left( \OutCMV{c}{\false}{\EncCMVmixCMV{S_3}} \mid J_5 \right)}, \quad \BranchCMV{\Label_{!}}{\left( \InpCMV{\!}{c}{z}{\EncCMVmixCMV{S_4}} \mid J_6 \right)}}}} \big)
			\end{array}}
	\end{align*}
	where $ J_7 $ again remains as junk.
	Then we perform a communication on $ xy $, where we chose the input on $ y $ in the first line:
	\begin{align*}
		T_2 \step T_3 = \ResCMV{x}{y}{\big(
			\!\!\begin{array}[t]{l}
				\ResCMV{c}{d}{\big(
					\!\!\begin{array}[t]{l}
						\BranCMV{c}{\Set{\BranchCMV{\textcolor{blue}{\Label_{?}}}{\left( \OutCMV{c}{\textcolor{blue}{\false}}{\textcolor{blue}{\EncCMVmixCMV{S_1}}} \mid J_1 \right)}, \quad \BranchCMV{\textcolor{darkgreen}{\Label_{!}}}{\left( \InpCMV{\!}{c}{\textcolor{darkgreen}{z}}{\textcolor{darkgreen}{\EncCMVmixCMV{S_2}}} \mid J_2 \right)}}}\\
						{}\mid \SelCMV{d}{\textcolor{orange}{\Label_{!}}}{\left( \OutCMV{d}{\textcolor{orange}{\true}}{\textcolor{orange}{\inactCMV}} \mid J_3\right)} \big) \mid J_7
					\end{array}}\\
				{}\mid \InpCMV{\!}{y}{c}{\BranCMV{c}{\Set{\BranchCMV{\Label_{?}}{\left( \OutCMV{c}{\false}{\EncCMVmixCMV{S_3}} \mid J_5 \right)}, \quad \BranchCMV{\Label_{!}}{\left( \InpCMV{\!}{c}{z}{\EncCMVmixCMV{S_4}} \mid J_6 \right)}}}} \big)
			\end{array}}
	\end{align*}
	Finally, two more steps on $ cd $ resolve the branching and transmit $ \textcolor{orange}{\true} $:
	\begin{align*}
		T_3 \step\step T_4 = \ResCMV{x}{y}{\big(
			\!\!\begin{array}[t]{l}
				\textcolor{darkgreen}{\EncCMVmixCMV{S_2}}\Set{\Subst{\textcolor{orange}{\true}}{\textcolor{darkgreen}{z}}} \mid J_2 \mid J_3 \mid J_7 \mid J_8\\
				{}\mid \InpCMV{\!}{y}{c}{\BranCMV{c}{\Set{\BranchCMV{\Label_{?}}{\left( \OutCMV{c}{\false}{\EncCMVmixCMV{S_3}} \mid J_5 \right)}, \quad \BranchCMV{\Label_{!}}{\left( \InpCMV{\!}{c}{z}{\EncCMVmixCMV{S_4}} \mid J_6 \right)}}}} \big)
			\end{array}}
	\end{align*}
	This completes the emulation of $ S \step S_2' $, \ie the emulation of the single source term step $ S \step S_2' $ required a sequence of target term steps $ \EncCMVmixCMV{S} \steps T_1 \step T_2 \step T_3 \step \step T_4 $.
	\qed
\end{example}

The operational soundness is defined in \cite{casalMordidoVasconcelos22} as 
(adapting the notation):
\begin{equation}
	\label{eq:cmvsound}
	\text{If $ \ArbitraryEncoding{S} \stepTarget T $ then $ S \stepSource S' $ and $ T \stepsTarget \asymp \ArbitraryEncoding{S'} $.}
\end{equation}
As visualised above, the encoding translates a single source term step into a sequence of target term steps.
Unfortunately, for such encodings the statement in (\ref{eq:cmvsound}) is not strong enough: 
with (\ref{eq:cmvsound}), we check only that the first step on a literal translation does not introduce new behaviour.
The requirement $ T \stepsTarget \asymp \ArbitraryEncoding{S'} $ additionally checks that the emulation started with $ \ArbitraryEncoding{S} \stepTarget T $ can be completed, but not that there are no alternative steps introducing new behaviour.
Hence we prove a correct version of soundness as defined in \cite{gorla10} (see Definition~\ref{def:goodEncodingA}).

\begin{lemma}[Soundness, $ \EncCMVmixCMV{\cdot} $]
	The encoding $ \EncCMVmixCMV{\cdot} $ is operationally sound modulo $ \bisimCMV $, \ie $ \EncCMVmixCMV{S} \steps T $ implies $ S \steps S' $ and $ T \steps \bisimCMV \EncCMVmixCMV{S'} $.
	\label{lem:soundnessEncCMVmixCMV}
\end{lemma}

As suggested we use $ \bisimCMV $, \ie a form of weak reduction barbed bisimilarity that we simply call bisimilarity in the following.
For soundness we have to show that all steps of encoded terms belong modulo bisimilarity to the emulation of a source term step.
To prove Lemma~\ref{lem:soundnessEncCMVmixCMV}, we analyse the sequence of steps $ \EncCMVmixCMV{S} \steps T $ and identify all source term steps $ S \steps S' $ whose emulation is started within $ \EncCMVmixCMV{S} \steps T $ and the target term steps $ T \steps \bisimCMV \EncCMVmixCMV{S'} $ that are necessary to complete all started emulations modulo bisimulation.
Therefore, we use an induction on the number of steps in the sequence $ \EncCMVmixCMV{S} \steps T $ and analyse the encoding function in order to distinguish between different kinds of target term steps and the emulations of source term steps to that they belong.
Note that, as it is typical for many encodability results, the proof of operational soundness is more elaborate than the proof of operational completeness presented in \cite{casalMordidoVasconcelos22}.

In Example~\ref{exa:translation} we have $ T_4 \bisimCMV \EncCMVmixCMV{S_2'} $, because all differences between $ T_4 $ and $ \EncCMVmixCMV{S_2'} $ are due to junk that cannot be observed modulo $ \bisimCMV $.
In fact, we have already $ T_3 \bisimCMV \EncCMVmixCMV{S_2'} $, since we consider a weak form of bisimulation here.

In the above variant of soundness $ T $ can catch up with the source term $ S' $ by the steps $ T \steps \bisimCMV \EncCMVmixCMV{S'} $.
This allows for so-called \emph{intermediate states}: target terms that are strictly in between the translation of two source terms, \ie $ T $ such that $ S \step S' $, $ \EncCMVmixCMV{S} \steps T \steps \bisimCMV \EncCMVmixCMV{S'} $, but neither $ \EncCMVmixCMV{S} \bisimCMV T $ nor $ \EncCMVmixCMV{S'} \bisimCMV T $ (see \cite{parrow1992, petersNestmannGoltz13}).
In $ \EncCMVmixCMV{\cdot} $ such intermediate states are caused by mapping the task of finding matching communication partners of a single source term step onto several steps in the target.
Consider the term $ T_2 $ in the above emulation of $ S \step S_2' $.
By picking the branch with label $ \Label_1 $, we discarded the branch with label $ \Label_2 $.
Because of that, the emulation starting with $ \EncCMVmixCMV{S} \steps T_2 $ can no longer emulate source term steps of $ S $ that use channel $ x $ for receiving, \ie $ T_2 \not\bisimCMV \EncCMVmixCMV{S} $.
But, since we have not yet decided whether we emulate a communication with the first or second choice on $ y $, we also have $ T_2 \not\bisimCMV \EncCMVmixCMV{S_2'} $ whenever $ S_2 \not\bisimCMVmix S_4 $.
Indeed, if we assume that $ S_1, S_2, S_3, S_4 $ are pairwise not bisimilar, then $ T_2 \not\bisimCMV \EncCMVmixCMV{S'} $ for all $ S \step S' $, \ie $ T_2 $ is an intermediate state.

The existence of intermediate states prevents us from using stronger versions of soundness, \ie with $ T \asymp \ArbitraryEncoding{S'} $ instead of the requirement $ T \stepsTarget \asymp \ArbitraryEncoding{S'} $ in soundness.
The encoding $ \EncCMVmixCMV{\cdot} $ needs the steps in $ T \steps \bisimCMV \EncCMVmixCMV{S'} $ to complete the emulation of source term steps started in $ \EncCMVmixCMV{S} \steps T $.
With the soundness result we can complete the proof of \cite{casalMordidoVasconcelos22} that $ \EncCMVmixCMV{\cdot} $ presented in
\cite[\S~7]{casalMordidoVasconcelos22} is good.

\begin{theorem}[Encoding from \CMVmix into \CMV]
	\label{thm:encodeCMVmixintoCMV}
	The encoding $ \EncCMVmixCMV{\cdot} $ from \CMVmix into \CMV presented in \cite{casalMordidoVasconcelos22} is good. By this encoding source terms in \CMVmix and their literal translations in \CMV are related by coupled similarity.
\end{theorem}
\noindent 
For the proof we take the completeness result from \cite{casalMordidoVasconcelos22} and our soundness result in Lemma~\ref{lem:soundnessEncCMVmixCMV}.
The proof of the remaining properties is simple.
That the combination of operational correspondence and barb sensitiveness induces a (weak reduction, barbed) coupled similarity that relates all source terms and their literal translations was proved in \cite{petersGlabbeek15}.
To obtain a tighter connection such as the bisimilarity, we would need the stronger version of soundness with $ T \asymp \ArbitraryEncoding{S'} $ instead of $ T \stepsTarget \asymp \ArbitraryEncoding{S'} $ (see \cite{petersGlabbeek15}).

As mentioned, a key feature of the encoding is to translate the nature of its summands, \ie whether they are send or receive actions, into the label used by the target term.
That this is possible, \ie that the prefixes for send and receive in a choice of \CMVmix can be translated to labels in a separate choice of \CMV such that the difference is not observable modulo the criteria in Definition~\ref{def:goodEncodingA}, gives us the last piece of evidence that we need.
\CMVmix does not allow to solve problems such as leader election (Theorem~\ref{thm:separateCMVmixfromPiviaLeaderElection}) that are standard problems for mixed choice; \CMVmix cannot express the synchronisation pattern \patternStar either that we associate with mixed choice (Theorem~\ref{thm:separateCMVmixfromPiviaStar}). 
Yet, \CMVmix can express the pattern \patternM which is associated with \emph{separate choice}, and is encoded by a language with only separate choice (Theorem~\ref{thm:encodeCMVmixintoCMV}).
We conclude that choice in \CMVmix is semantically rather a separate choice.

\begin{corollary}
	\label{col:separateChoiceCMVmix}
	$ $\\
	The extension of \CMV given by \CMVmix introduces a form of separate choice rather than mixed choice.
\end{corollary}


\section{Related Work and Outlook}
\label{sec:conclusions}

We conclude by discussing related work, summing up our results, and briefly discussing our next steps.

\subsection{Related Work}

Encodings or the proof of their absence are the main way to compare process calculi \cite{boerPalamidessi1991, parrow08, gorla10, fuLu10, peters12, vG12, parrow2014abstraction, fu16, glabbeek18}.
See \cite{peters19} for an overview and discussion on encodings.
We used this methodology to compare different variants of choice in
session types.

The relevance of mixed choice for the expressive power of the \piCal was extensively studied.
An important encodability result on choices 
is the existence of a good encoding from the choice-free synchronous \piCal into its asynchronous variant \cite{boudol92, honda.tokoro:objectcalculus}, since it proves the relevance of choice.
As for the separation result, 
\cite{palamidessi03, gorla10, breakingSymmetries16} have shown that there is no good encoding from the full \piCal, \ie the synchronous \piCal including mixed choice, into its asynchronous variant if an encoding should preserve the distribution of systems.
Palamidessi in \cite{palamidessi97} was the first to point out that mixed choice strictly raises the expressive power of the \piCal.
Later work studies the criteria under that this separation result
holds and alternative ways to prove this result:
\cite{nestmann00} studies the relevance of divergence reflection for this result and considers separate choice.
\cite{gorla10, parrow08} discuss how to reprove this result if the rather strict criterion on the homomorphic translation of the parallel operator is replaced by compositionality.
\cite{peters12, petersNestmann12} show that compositionality itself is not strong enough to replace the homomorphic translation of the parallel operator by presenting an encoding and then propose the preservation of distributability as criterion to regain the result of Palamidessi.
\cite{breakingSymmetries16} uses the more fundamental problem of breaking symmetries instead of leader election.
\cite{petersNestmannGoltz13} further simplifies this separation result by introducing synchronisation patterns to distinguish the languages.
\cite{peters16} shows that instead of the preservation of distributability or the homomorphic translation of the parallel operator also the preservation of causality can be used as criterion.

While there are a vast amount of
theories~\cite{DBLP:journals/csur/HuttelLVCCDMPRT16}, 
programming languages 
~\cite{DBLP:journals/ftpl/AnconaBB0CDGGGH16}, 
and tools~\cite{BETTYTOOLBOOK} of session types, 
as far as we know, the \CMVmix-calculus  
is the only session $\pi$-calculus which extends external and internal choices to 
their mixtures with full constructs, 
i.e.~delegation, shared (or unlimited) name passing, 
value passing, and recursion in its process syntax,   
proposes its typing system and proves type-safety. 
In the context of \emph{multiparty session types} \cite{HYC2016}, 
there are several works that extend 
the original form of global types where choice is fixed (from one sender to one receiver) with more flexible forms of choices:
Recent work in \cite{DBLP:conf/concur/MajumdarMSZ21} \eg allows the global type to specify a choice of one sender to transmit to one of several receivers.
In \cite{JY-ESOP20} flexible choices are discussed but their well-formedness (which ensures deadlock-freedom of local types) needs to be checked by bisimuluation.
These works focus on gaining expressiveness of behaviours of a set of local types (or a simple form of CCS-like processes which are equivalent to local types \cite{DBLP:conf/concur/MajumdarMSZ21}) which correspond to \emph{a single multiparty session}, without delegations, interleaved sessions, restrictions nor name passing.

More recently, \cite{DBLP:conf/esop/Glabbeek22} compares the expressive power of a variant of the \piCal (with implicit matching) and the variant of CCS where the result of a synchronisation of two actions is itself an action subject to relabelling or restriction.
Because of the connection between CCS-like languages and local types, it may be interesting to compare the expressiveness results in \cite{DBLP:conf/esop/Glabbeek22} with (variants of) multiparty session types.

\subsection{Summary and Outlook}

We proved that \CMVmix is strictly less expressive than the \piCal in two different ways: by showing that \CMVmix cannot solve leader election in symmetric networks of odd degree and that \CMVmix cannot express the synchronisation pattern \patternStar.
Then we provide the missing soundness proof for the encoding presented in \cite{casalMordidoVasconcelos22}.
From these results and the insights on the reasons of these results, we conclude that the choice primitive added to \CMV in \cite{casalMordidoVasconcelos22} is rather a separate choice and not a mixed choice at least with respect to its expressive power.

\begin{figure}
	\centering
	\begin{tikzpicture}[node distance=3cm, auto]
		\node (mix)		at (0, 1.8)		{\piMix};
		\node (sep)		at (-2, 0)	{\piSep};
		\node (asyn)	at (-1.25, 0)		{\piAsyn};
		\node (ma)		at (-0.3, 0)		{\MA};
		\node[color=blue] (CMV)		at (0.75, 0)		{\CMV};
		\node[color=blue] (CMVmix)		at (2, 0)		{\CMVmix};
		\node (join)	at (-0.75, -1.8)	{\join};
		\node (mau)		at (0.75, -1.8)	{\MAu};

		\node[scale=2] (star) at (2.75, 1) {\ \patternStar};
		\draw[dashed] (-2, 1) -- (2,1) -- (3,2);
		\draw[dashed] (2,1) -- (3,0);

		\node[scale=1.2] (M) at (-2.75, -1) {\patternM\quad};
		\draw[dashed] (2,-1) -- (-2,-1) -- (-3,-2);
		\draw[dashed] (-2,-1) -- (-3,0);
	\end{tikzpicture}
	\caption{Hierarchy of Pi-like Calculi.}
	\label{fig:hierarchy}
\end{figure}
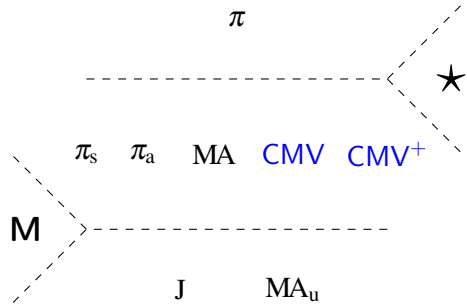

With these results we can extend the hierarchy of pi-like calculi obtained in \cite{petersNestmannGoltz13, petersNestmannIC20} by two more languages as depicted in Figure~\ref{fig:hierarchy}.
This hierarchy orders languages according to their ability to express certain synchronisation patterns.
At the top we have the \piCal ($ \piMix $), because it can express the synchronisation pattern \patternStar.
In the middle are languages that can express \patternM but not \patternStar: the \piCal with separate choice ($ \piSep $) \cite{nestmann00}, the asynchronous \piCal without choice ($ \piAsyn $) \cite{honda.tokoro:objectcalculus, boudol92}, Mobile Ambients ($ \MA $) \cite{cardelliGordon00}, \CMV, and \CMVmix.
In the bottom we have the join-calculus ($ \join $) \cite{fournet.gonthier:join} and Mobile Ambients with unique Ambient names ($ \MAu $) \cite{petersNestmannIC20}, \ie the languages that cannot express \patternStar or \patternM.
That $ \piMix, \piSep, \piAsyn, \MA, \join $, and $ \MAu $ can or cannot express the respective pattern was shown in \cite{petersNestmannGoltz13, petersNestmannIC20}.

Linearity as enforced by the type system of \CMV/\CMVmix restricts the possible structures of communication protocols.
In particular, the type system ensures that it is impossible to unguard two competing inputs or outputs on the same linear channel at the same time.
Accordingly, it is not surprising that adding choice, even mixed choice, towards communication primitives under a type discipline that enforces linearity does not significantly increase the expressive power of the respective language (though it still might increase flexibility).
However, that adding mixed choice between unrestricted communication primitives does not significantly increase the expressive power of the language, did surprise us.
Unrestricted channels allow to have several in- or outputs on these channels in parallel, because the type system only ensures the absence of certain communication mismatches as \eg that the sort of a transmitted value is as expected by the receiver; but not linearity (compare also to shared channels as \eg in \cite{hondaVasconcelosKubo98}).
So, there is no obvious reason why the type system should limit the expressive power of unrestricted channels within a mixed choice.
Indeed, it turns out that the problem lies not in the type system.
In both ways to prove the separation result in \S~\ref{sec:separateMixedSessionsLeaderElection} and \S~\ref{sec:separateMixedSessionsFromPiSynchronisationPatterns} we completely ignore the type system and carry out the proof on the untyped version of the language, \ie it is already the untyped version of \CMVmix that cannot express mixed choice despite a mixed-choice-like primitive.
This limitation of the language definition, \ie in its syntax and semantics, is not obvious and indeed it was very hard to spot the problem.

We expect that adding mixed choice to the non-linear parts of other session type systems will instead significantly increase the expressive power. Accordingly, as the next step, we want to add a primitive for mixed choice between shared channels in session types such as described \eg in \cite{hondaVasconcelosKubo98, yoshidaVasconcelos06} and analyse the expressiveness of the resulting language.

\paragraph{Acknowledgements.}
The work is partially supported by EPSRC (EP/T006544/1, EP/K011715/1, \linebreak EP/K034413/1, EP/L00058X/1, EP/N027833/1, EP/N028201/1, EP/T006544/1, EP/T014709/1, \linebreak EP/V000462/1 and EP/X015955/1) and NCSS/EPSRC VeTSS.


\bibliography{references}

\end{document}